%% file: main.tex
\newif\ifnotes
\notestrue

\documentclass[12pt]{article}
\ifnotes 
    \usepackage[draft=false, notes=true, marginpar=false, done=false, later=false]{dtrt}
\else
    \usepackage[draft=false, notes=xxx, marginpar=false, done=false, later=false]{dtrt}
\fi

\input{macros}

\input{packages}

\makeatletter
\renewcommand{\dt@brokenrefmark}[1]{}
\makeatother
\usepackage[top=1in, bottom=1in, left=1in, right=1in]{geometry}

\let\cite\citeyear

\title{TradeMech: A Method to Multilaterally Net Trades Without Altering Counterparty Exposure}
\author{Daniel Aronoff\thanks{Research Affiliate, MIT Department of Economics}, Robert Townsend\thanks{Elizabeth and James Killian Professor of Economics, MIT Department of Economics} and Madars Virza \thanks{Chief Scientist, Radius} }

\date{November 11, 2026}

\begin{document}

\maketitle

\pagenumbering{roman}
\begin{abstract}
\input{abstract}
\end{abstract}

\pagebreak
\tableofcontents
\pagebreak

\newpage
\pagenumbering{arabic}
\onehalfspacing

\input{introduction}

\input{assignment_of_trades_to_chains}

\input{preserving_counterparty_risk}

\input{additional_applications}

\input{conclusion}

\bibliography{references}

\end{document}

%% file: macros.tex
\usepackage[framemethod=TikZ]{mdframed}
\mdfdefinestyle{MyFrame}{%
    linecolor=blue,
    outerlinewidth=1pt,
    roundcorner=10pt,
    innertopmargin=\baselineskip,
    innerbottommargin=\baselineskip,
    innerrightmargin=20pt,
    innerleftmargin=20pt,
    backgroundcolor=white}

\newcommand{\m}{$M$\xspace}
\newcommand{\g}{$T$\xspace}
\newcommand{\mo}{mechanism operator\xspace}
\newcommand{\mpc}{multiparty contract\xspace}
\newcommand{\ic}{initial contract\xspace}
\newcommand{\tm}{trading mechanism\xspace}

%% file: packages.tex
\usepackage[utf8]{inputenc}

\usepackage[english]{babel}
\usepackage{amsthm}
\usepackage{bbm}
\usepackage{xcolor}
\usepackage{longtable}
\usepackage{booktabs, caption, makecell}

\usepackage{threeparttable}
\DeclareMathAlphabet{\mathpzc}{OT1}{pzc}{f}{it}
\usepackage{appendix}

\theoremstyle{plain}
\usepackage{graphicx}
\usepackage{ulem}
\usepackage{xspace}
\usepackage{tikz}
\usetikzlibrary{arrows.meta}
\usepackage{tikz,soul,setspace}

\usepackage{multicol}

\graphicspath{{./}} %

\usepackage{float}
\usepackage{parskip}%
\usepackage{adjustbox}
\usepackage{lipsum}
\usepackage{tabularx}

\newtheorem*{proposition}{Proposition 1}
\newtheorem*{proposition*}{Proposition 2}

\newtheorem*{theorem*}{Proposition 1: Edge and Flow Preservation}
\newtheorem*{theorem**}{Theorem 8.2}
\newtheorem*{theorem***}{Theorem: A Computational Bound}
\newtheorem*{corollary*}{Corollary 1}
\newtheorem*{corollary**}{Corollary 2}
\newtheorem*{corollary***}{Corollary}
\newtheorem*{definition*}{Definition 1: Chains}
\newtheorem*{definition**}{The Replacement Repo Contract Guarantee Structure}
\newtheorem*{definition***}{Definition 2: Cycles}

\newtheorem{lemma}{Lemma}
\newtheorem*{lemma*}{Lemma: Equal Flow Across Cuts}
\newtheorem*{lemma**}{Lemma: Flow Decomposition}

\newtheorem*{definition****}{Definition: Maximal Netting}

\usepackage{multirow}

\usepackage{float}

\usepackage{import}
\usepackage{marginnote}
\usepackage{changepage}

\usepackage{titlesec}
\usepackage{inputenc}
\usepackage{csquotes}
\usepackage{endnotes}
\usepackage{natbib}

\usepackage{marvosym,soul,graphicx,%
xcolor,endnotes,amssymb,amsmath,amsthm,setspace,soul}
\definecolor{rune}{HTML}{4A6672}
\newcommand{\rune}[1]{\textcolor{rune}{#1}}
\usepackage{amsmath,amssymb}
\usepackage[ruled,vlined,linesnumbered]{algorithm2e}
\SetKwInOut{Input}{Input}
\SetKwInOut{Output}{Output}

\hypersetup{colorlinks=true,linkcolor=black,urlcolor=black,citecolor=black}

\usepackage[many]{tcolorbox}   
\definecolor{main}{HTML}{000000}    %
\definecolor{sub}{HTML}{FFFFFF}     %

\newtcolorbox{boxH}{
    colback = sub, 
    colframe = main, 
    boxrule = 1pt, 
    leftrule = 1pt %
}

%% file: abstract.tex
Financial markets, such as markets for repo, bonds and derivatives, create a network of interdependent obligations. Various netting techniques are used to consolidate multiple offsetting obligations. Existing techniques present a tradeoff between the extent of netting and the preservation of counterparty risk. More of one implies less of the other. TradeMech is a mechanism that achieves maximal multilateral netting of trades without altering counterparty risk. Initial contract trades are assigned to chains and cycles and then terminated and replaced by multilateral contracts that net trades on each chain and cycle. When an agent fails to pre-commit its required traded objects, its trade assignment is recovered, removed from the chain or cycle and instantiated with a bilateral contract. This preserves counterparty risk. The remaining trade assignments are netted on new chains.

\bigskip
\bigskip
{\small JEL Classifications: \rune{D24},
\rune{D47},
\rune{D82},\rune{D86},
\rune{G12},
\rune{G21},
\rune{G23},
\rune{G24},
\rune{G28}}

\thispagestyle{empty}

%% file: introduction.tex
\section{Introduction}

Financial markets, such as markets for bonds, derivatives, and repo, form a network of interdependent obligations. To reduce settlement complexity and counterparty exposure, various netting techniques are used to consolidate multiple offsetting obligations. Existing methods present a tradeoff between the extent of netting and the preservation of counterparty risk---the risk of counterparty default---in the trade agreements initially entered into by counterparties. More of one implies less of the other. TradeMech is a mechanism, applicable to markets in which one or two homogeneous fungible objects are traded, that achieves maximal multilateral netting of the designated netted object without altering the location of counterparty risk. The mechanism transforms the network formed by the initial contracts into a set of chains and cycles on which traded objects flow along edges between nodes. The initial contracts are terminated and replaced by multiparty contracts that net trades on each chain and cycle while preserving each agent's contractual profit. When a node fails to pre-commit the objects it is required to send, the affected trade assignment is removed from the chain or cycle, recovered as a bilateral contract between the same original counterparties, and the remaining trade assignments are re-netted on residual chains. No new counterparty exposures are created and no initial counterparty exposures are lost. Thus, the mechanism has three key properties: it implements an ex ante algorithmic rearrangement of contracts, achieves maximal multilateral netting, and preserves the location of counterparty risk.

\subsection{Alternative trading mechanisms}
The two prevalent methods which have been implemented in financial markets that achieve multilateral netting are central clearing and trade compression.
Central clearing is used, inter alia, in the US Treasuries secondary and repo markets (SEC \cite{SEC-UST-CCP-Rule-2023}) and in OTC markets for credit default swaps, interest rate derivatives and FX derivatives \citep{BIS_OTCStats_2023,BIStradecompression2019}. Trade compression is used, inter alia, in OTC markets for interest rate derivatives, FX and equities (IOSCO \cite{IOSCO_PTRRS_2024}).

\textbf{Central clearing} In central clearing, an entity interposes itself between trade counterparties – becoming the buyer to every seller and the seller to every buyer (the ''CCP''). This transforms a network of bilateral exposures into a hub-and-spoke model where each participant only has exposure to the CCP. For example, if $i$ owes $j$ \$1 and $j$ owes $k$ \$1, central clearing would interpose the CCP would eliminate $j$'s involvement and interpose between $i$ and $k$. $i$ wold owe the CCP \$1 and the CCP would owe \$1 to $k$. The CCP induces maximal multilateral netting, but counterparty credit risk has shifted to the CCP. Moreover, the CCP may add cost to the trades. A CCP mitigates bilateral default risk through robust margining and default fund mechanisms. Losses are absorbed by the CCP (and potentially mutualized among members) rather than by the originally contracting counterparty. This concentrates risk in the CCP and severs the direct risk link between original trading partners as agents no longer manage credit risk bilaterally and must instead rely on the CCP’s solvency and risk sharing arrangements. 

\textbf{Trade compression} Trade compression is a form of multilateral netting often used in over-the-counter (OTC) derivatives markets to tear up redundant contracts. Services like TriOptima’s TriReduce allow multiple counterparties to collectively cancel offsetting trades, reducing gross notional outstanding. In such multilateral compression cycles, participants contribute their portfolios into an algorithm that identifies offsetting positions and net flows. For example, if $i$ owes $j$ \$1 and $j$ owes $k$ \$1, trade compression would eliminate $j$'s participation and create a direct obligation from $i$ to $k$. This would transfer $k$'s counterparty credit risk from $i$ to $j$. Moreover, since the acceptance of the new arrangement requires the approval of $i$ and $k$, there is no guarantee that trade compression would achieve maximal netting of trades.\footnote{``[C]ompression cycles are not always fully efficient in offsetting all possible positions.'' BIS \cite{BIStradecompression2019}.} Therefore, trade compression neither preserves counterparty credit risk nor achieves maximal multilateral netting.\footnote{ A recent study explores applying trade compression techniques to offset trade credits (Boissay et.al. \cite{tradecredit}).}

\subsection{Related literature}

Our study relates to the literature on market mechanisms that achieve multilateral netting of trades. One strand studies central clearing and the tradeoff between multilateral netting benefits and the reallocation of counterparty exposure. Duffie and Zhu \cite{Duffie-Zhu} provide the foundational analysis of whether central clearing reduces counterparty risk, while Duffie \cite{Duffie2017,Duffie2020} Group of Thirty \cite{Group-of-30,G30-2} and the SEC \cite{SEC-UST-CCP-Rule-2023} study or implement central-clearing reforms in U.S. Treasury markets. A second strand studies post-trade compression and related post-trade risk reduction services in OTC markets. Ehlers and Hardy \cite{BIStradecompression2019} describe the evolution of OTC derivatives market structure, including compression and clearing, and IOSCO \cite{IOSCO_PTRRS_2024} reviews post-trade risk reduction services and sound practices. Boissay et al. \cite{tradecredit} analyze multilateral trade-credit setoff, which is adjacent to our setting. We contribute to this literature by proposing a mechanism that achieves maximal multilateral netting without reallocating counterparty exposure.

\subsection{Roadmap}

In section \ref{sec:Assignment of Trades to Chains and Cycles} we describe a protocol -- composes of a sequence of algorithms -- that decomposes the aggregate of initial bilateral contracts for the exchange of two homogeneous objects into multiparty contracts on chains and cycles. Proposition 1 demonstrates that first result; that the protocol achieves maximal multilateral netting and minimizes the net expenditure of money. In Section \ref{sec:preserving_counterparty_risk} we explain how the mechanism preserves initial counterparty risk by removing the trade assignment of a party that fails to pre-commit the objects it is required to send on a chain and cycle and re-netting the remaining trade assignments. Proposition 2 demonstrates that Trademech preserves profit and counterparty risk. Algorithm 3 summarizes the protocol for iterating over chains and cycles to find the executable netting contracts. Section \ref{sec:Single_Traded_Object} extends the protocol to the case of single object. We conclude with a brief overview of the mechanism.

%% file: assignment_of_trades_to_chains.tex
\section{Assignment of Trades to Chains and Cycles}
\label{sec:Assignment of Trades to Chains and Cycles}

Trademech terminates initial  contracts, re-allocates trades to chains and cycles and implements replacement contracts on the chains and cycles.  In this section we describe the protocol for assigning initial  contract trades to chains and cycles. Each initial   contract between two agents involves an exchange of one object, which we denote \m, for another object, \g. This imposes a choice as to which object to net. WLOG  we net \g. The protocol decomposes the flow of \g between agents, represented as nodes, into a set of chains and cycles. It then  re-attaches the initial contract \m flows to the \g flows between pairs of nodes. 

The reshaping of  trade flows is accomplished in five steps. The first step nets trades between initial contract counterparties. It generates a graph where nodes represent agents and directed edges represent initial  contract net \g flows between agents (Section \ref{subsec: Network structure formed by initial trade contracts}). The second step splits an agent's node into two; one representing its excess inflow or outflow of \g and the other representing the remaining matched-trade equal inflow and outflow. We call these nodes "children" of the agents (Section \ref{subsec: Separate excess trade flows and matched-trade trade flows}). The third step transforms the graph into a flow network (Section \ref{subsec: Transform T flows into a flow network}). The fourth step decomposes the flow network into a set of chains and cycles, with nodes representing agents and directed edges representing initial  contract \g flows. We call these nodes the "grandchildren" of the agents (Section \ref{subsec: Transform the TFN into chains and cycles}). The fifth step re-attaches the flows of \m from the initial  contracts to the chains and cycles. The edges are bi-directional, which \m flowing in one direction and \g flowing in the other direction (Section \ref{subsec:Attaching  M to chains and cycles}).

\subsection{Network structure formed by initial trade contracts}
\label{subsec: Network structure formed by initial trade contracts}

Table \ref{tab:Initial Contracts} displays initial contracts between agents. Each row is an initial contract trade. The first column is the initial contract number. The second column is the agent that sends \g (and receives \m). The third column is the counterparty agent that sends \m (and receives \g). The fourth column is the unit price of \g in terms of \m. The fifth column is the traded units of \g.

\begin{table}[h!]
\centering
\caption{Initial Contracts}
\label{tab:Initial Contracts}
\begin{tabular}{lllll}
\toprule
\makecell{\textbf{Contract}\\\textbf{Number}} &\makecell{\textbf{Firm}\\\textbf{ Id}} & \makecell{\textbf{Counter -}\\\textbf{ party Id}} & \makecell{\textbf{Unit Price}\\\textbf{ \g for \m}} & \makecell{\textbf{Units of \g}}\\
\midrule
1 & h & i & 5.25 & 5  \\
2 & k & i & 6.3  & 3  \\
3 & i & j & 6.55 & 5  \\
4 & i & g  & 3    & 4  \\
5 & g & j & 5.95 & 10 \\
6 & l & g & 5.95 & 6  \\
7 & h & f & 3.3  & 10 \\
8 & f & h & 3.1  & 8  \\
9 & k & g & 3.77 & 8  \\
10 & g & f & 6.53 & 10 \\
11 & f & i & 5.12 & 6  \\
\bottomrule
\end{tabular}
\end{table}

After the initial contracts are executed, trades between  agents are netted on object \g. The netting works as follows. In contract \#7 agent $h$ sends 10 units of \g to agent $f$. In contract \#8 agent $f$ sends 8 units of \g to agent $h$. In the net trade agent $h$ sends agent $f$ 2 units of \g. The total \m paid by $f$ in contract \#7 is 30.3. The total \m paid by $h$ in contract \#8 is 24.8. The net price paid by $f$ to $h$ is 5.5. The net unit price is 2.75. Figure \ref{fig:Initial S Flow Network} depicts the network of \g flow formed by the netted initial contracts. Nodes represent agents and the numbers on the directed edges represent the units of \g flowing between the agents.\footnote{ We sometimes represent the flow from e.g. $h$ to $f$ by  $2_{h\to f}$ or, abstractly as $\g_{h\to f}$.}

\begin{figure}[H]
\begin{center}
\includegraphics[page=1,width=346pt,height=250pt,keepaspectratio]{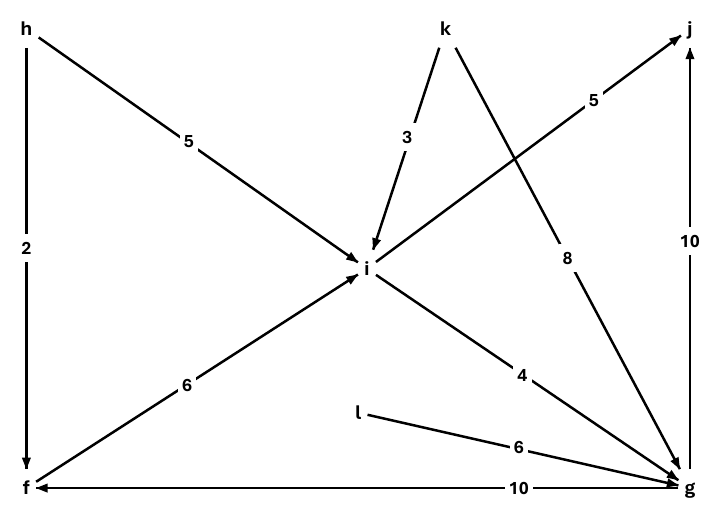}
\end{center}
\caption{Initial  \g - Flow Network}
\label{fig:Initial S Flow Network}
\end{figure}

\subsection{Separate excess trade flows and matched-trade trade flows}         
\label{subsec: Separate excess trade flows and matched-trade trade flows}

The second step divides each node into at most two nodes; one node has equal inflow and outflow of \g (the "balanced node" or "matched-trade node") and the other node has the excess of inflow or outflow (the "excess flow node"), if any. Figure \ref{fig:Splitting node f with a net inflow of 6} first splits parent node $f$, which has a net inflow of 6, into the balanced-trade node $BT_{f}$ in the middle and the excess-inflow node $NR_{f}$ on the right. Figure \ref{fig:Splitting node g with a net outflow of 2} then starts from that updated network and splits parent node $g$, which has a net outflow of 2, into the balanced-trade node $BT_{g}$ below $BT_{f}$ and the excess-outflow node $NS_{g}$ on the left. The split of $f$ is retained while $g$ is split. The original $f$ and $g$ nodes are the "parent" nodes, and the balanced and excess flow nodes into which they are divided are their "child" nodes. A key property of the node split is that the aggregate flows between the child nodes of $f$ and $g$ are the same as the flow between the parent $f$ and $g$ in their netted initial contracts.

Two notable feature are; first, the trade pattern for the children of $f$ and $g$ is unaffected by the order in which the splitting occurs. Second, the lowest value trades, in terms of the unit price of \g in exchange for \m, were moved to the excess flow nodes. This can be read from Table \ref{tab:Initial Contracts}.\footnote{More generally, any criteria for selecting the trade flows assigned to the excess trade nodes would work.} These features are formally stated below.

\begin{itemize}
\item The order in which the parent nodes are split does not affect a node's excess \g flow. For example, $g$ sends $f$ 10 units of \g in the initial \g-flow network (Figure \ref{fig:Initial S Flow Network}). After $f$ is split, 4 units flow from $g$ to $BT_{f}$ and 6 units from $g$ to $NR_{f}$. Splitting $g$ then assigns 2 units to $NS_{g}\to BT_{f}$, 2 units to $BT_{g}\to BT_{f}$, and 6 units to $BT_{g}\to NR_{f}$ (Figure \ref{fig:Splitting node g with a net outflow of 2}). The aggregate flow from $g$ to $f$ remains 10 units.

\item A node's excess \m flow is affected by the selection of \g flows it re-routes to its excess flow child node. Algorithm \ref{algo:node-splitting} minimizes excess \m flows by re-routing \g flows in ascending order of the size of the associated \m flow. 

\end{itemize}

\begin{figure}[H]
\centering
\includegraphics[page=1,width=346pt,height=250pt,keepaspectratio]{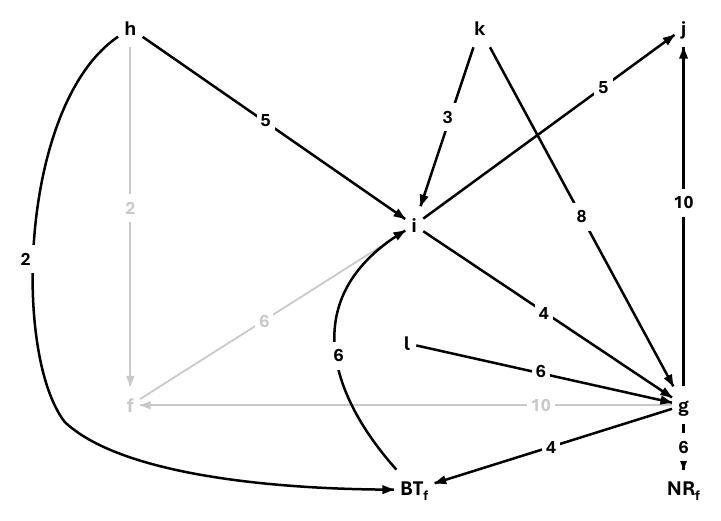}

\caption{Splitting node $f$ with a net inflow of 6}
\label{fig:Splitting node f with a net inflow of 6}
\end{figure}

\begin{figure}[H]
\centering
\includegraphics[page=1,width=346pt,height=250pt,keepaspectratio]{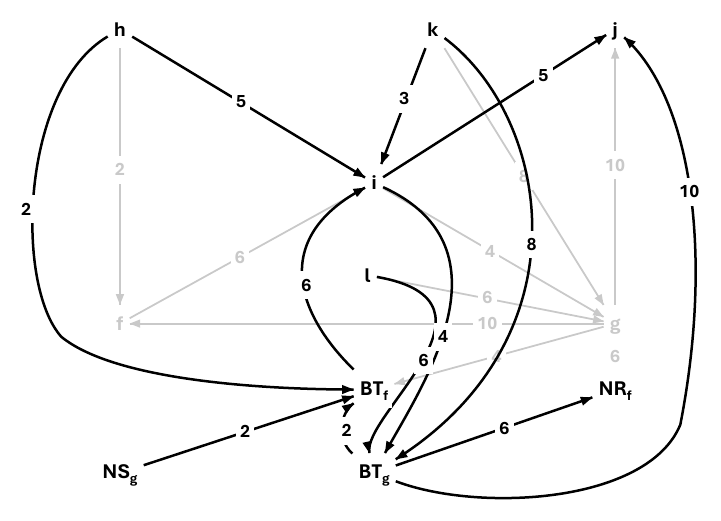}

\caption{Splitting node $g$ with a net outflow of 2}
\label{fig:Splitting node g with a net outflow of 2}
\end{figure}

Algorithm 1 states the protocol for splitting nodes.

\begin{algorithm}[H]
\caption{Node Splitting Algorithm}
\label{algo:node-splitting}
\Input{Directed trade-flow graph $R=(A,\mathcal{E})$ with edge flows $f:\mathcal{E}\to\mathbb{R}_{\ge 0}$ and unit prices $p:\mathcal{E}\to\mathbb{R}_{\ge 0}$; node $i\in A$ with net imbalance $\gamma_i$ and magnitude $E \gets |\gamma_i|>0$}
\Output{Updated graph $R'$ in which $i$ is replaced by $BT_{X}$ and either $NR_{X}$ (outflow case) or $NS_{X}$ (inflow case)}
$R' \gets R$\;

\If{$\gamma_i > 0$}{ %
  \tcp{Net outflow $E$ at node $i$}
  $S,\ (x_e)_{e\in S} \gets \operatorname{SelectOutflowByPrice}(i,E,p)$ \tcp{choose $E$ units from outgoing edges in ascending unit price; $\sum_{e\in S} x_e = E$}
  \ForEach{$e \in S$}{
    $f[e] \gets f[e] - x_e$\;
  }
  $NR_{X} \gets \operatorname{NewNode}()$\;
  $R' \gets \operatorname{ReassignOutflow}(R', i, NR_{X}, \{(e,x_e)\}_{e\in S})$\;
  $R' \gets \operatorname{Relabel}(R', i, BT_{X})$\;
}
\ElseIf{$\gamma_i < 0$}{ %
  \tcp{Net inflow $E$ at node $i$}
  $S,\ (x_e)_{e\in S} \gets \operatorname{SelectInflowByPrice}(i,E,p)$ \tcp{choose $E$ units from incoming edges in ascending unit price; $\sum_{e\in S} x_e = E$}
  \ForEach{$e \in S$}{
    $f[e] \gets f[e] - x_e$\;
  }
  $NS_{X} \gets \operatorname{NewNode}()$\;
  $R' \gets \operatorname{ReassignInflow}(R', i, NS_{X}, \{(e,x_e)\}_{e\in S})$\;
  $R' \gets \operatorname{Relabel}(R', i, BT_{X})$\;
}
\Else{
  \tcp{No split needed when $\gamma_i=0$}
}
\Return{$R'$}\;
\end{algorithm}

\subsection{Transform  \textit{T}-flows into a flow network}
\label{subsec: Transform  T flows into a flow network}

The third step adds a source node $s$ and a sink node $t$ after all parent nodes have been split. Figure \ref{fig:Trade Flow Network} completes the splits initiated at $f$ in Figure \ref{fig:Splitting node f with a net inflow of 6} and continued at $g$ in Figure \ref{fig:Splitting node g with a net outflow of 2}, preserving their child nodes and flows while splitting the remaining parent nodes. The result is the trade flow network ("TFN") for the market transactions in Table \ref{tab:Initial Contracts}. Each artificial edge from $s$ to an $NS$ node and from an $NR$ node to $t$ has unlimited capacity. The $NS$ nodes send \g (and only receive \g from $s$); the $BT$ nodes receive and send equal amounts of \g; and the $NR$ nodes receive \g (and only send \g to $t$). Removing the artificial edges from an $s$--$t$ path leaves a chain from an $NS$ node to an $NR$ node.

\begin{figure}[H]
\begin{center}
\includegraphics[page=1,width=346pt,height=250pt,keepaspectratio]{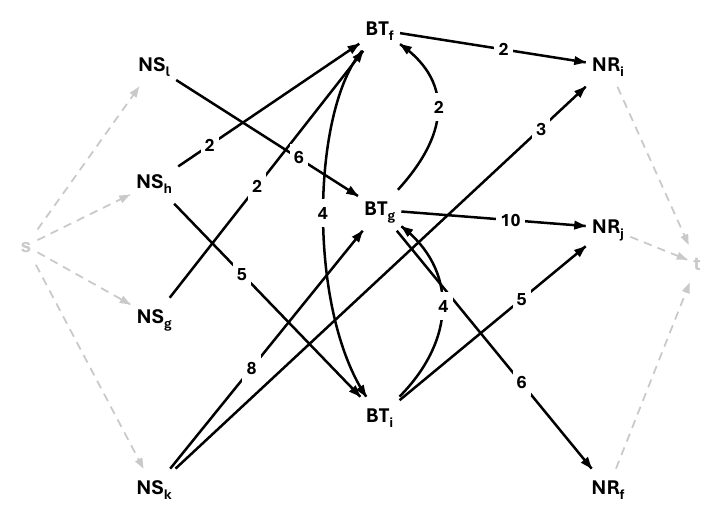}
\end{center}
\caption{Final trade flow network (TFN) after all node splits, continuing the $f$-then-$g$ sequence in Figures \ref{fig:Splitting node f with a net inflow of 6} and \ref{fig:Splitting node g with a net outflow of 2}.}
\label{fig:Trade Flow Network}
\end{figure}

\subsubsection{Computational complexity of creating the TFN} 

The initial \g - flow network (Figure \ref{fig:Initial S Flow Network}) is transformed into the TFN in three steps.  First, the net flow of a chosen node  is computed. Second, the flow to and from the $V-1$ other nodes are sorted (denoting $V$ the number of nodes in the initial \g - flow network). Third, flow is reallocated and edges are attached to the new $NR$ or $NS$ node. Equation \ref{eq:Comp cmplexity of TFN} is the upper bound complexity to split a node.

\begin{equation}
\label{eq:Comp cmplexity of TFN}
\underset{\text{inflow-outflow}}{\underbrace{1}} + \underset{\text{sorting}}{\underbrace{ \mathcal{O}(V - 1)\text{log}(V - 1))}} + \underset{\text{new node and edge}}{\underbrace{ \mathcal{O}(1)}}
\end{equation}

Summing over $V$ nodes we obtain complexity $\mathcal{O}(V^{2}\text{log} V)$.\footnote{The result uses that (i) when $V$ is large $(V+1)\sim V$ and (ii) $\mathcal{O}(V - 1)\text{log}(V - 1)$ dominates $\mathcal{O}(V^{2})$ and $\mathcal{O}(V)$.}

\subsubsection{Properties of the TFN}
\label{subsec: The Maximum Flow problem}

The flows of \g in the TFN form a network $R = (A,E)$ where nodes $u, v\in A = \{1,...,A\}$ denote agents; directed edges $(u,v) \in E = \{1,...,E\}$ have capacity equal to the volume $\g_{u\rightarrow v}$ of  \g between nodes $u$ and $v$ that is allocated to the edge.\footnote{We denote nodes as $\{u, v\}$ rather than $\{i,j\}$ to underscore that nodes in the TFN are not the same as nodes in the initial  \g-flow graph \ref{fig:Initial S Flow Network}}. A $flow$ in $R$ is a real valued function $f:\;A\times A \rightarrow \mathbb{R}$, which represents the flow of \g between a pair of nodes that are connected by an edge. A flow network satisfies the following two properties:\footnote{See Cormen et.al. \cite{cormen-introduction} Chapter 26.1.}

\textit{Flow Conditions}

\textbf{Capacity constraint}: For all $u, v \in A$, $0 \leq f(u, v) \leq$  $\g_{u\rightarrow v}$

\textbf{Skew symmetry}: For all $u, v \in A$, $f(u, v) = -f(v, u)$.

\textbf{Flow conservation}: For all $u \in A - \{s,t\}$,

$\underset{\text{inflows of \g}}{\underbrace{\sum_{u\in A}f(u, v)}} = \underset{\text{outflows of \g}}{\underbrace{\sum_{u\in A}f(v,u)}}$

When $(u,v) \notin E$, $f(u,v) = 0$

It is immediate that the TFN is a flow network. Noting that node $s$ sends flow to nodes labeled $NR$, which in turn send flow to nodes labeled $BT$, the measurement of network \textit{flow} $|f|$ is the volume of \g sent from $NR$ nodes. The formula for network flow is the following.

$|f| = \sum_{u\in A}f(s,u) = \sum_{NS_{i}\in A}\sum_{BT_{j}\in A}f(NS_{i},BT_{j})$

A \textit{cut} $(\mathbb{S},\mathbb{T})$ of the TFN is a partition of nodes into two sets; $S$ and $\mathbb{T} = A - \mathbb{S}$ such that $s\in \mathbb{S}$ and $t\in \mathbb{T}$. The flow $f(\mathbb{S},\mathbb{T})$ across the cut $(\mathbb{S},\mathbb{T})$ is defined to be

$f(\mathbb{S},\mathbb{T}) = \sum_{u\in\mathbb{S}}\sum_{v\in \mathbb{T}}f(u,v)$

Since the outflow of the $NS$ nodes to $BT$ nodes equals the inflow received from $s$, the network flow in the TFN is equal to the volume of \g flowing across the cut that has $s$ and the $NS$ nodes in $\mathbb{S}$ and all other nodes in $\mathbb{T}$. Cormen et.al. \cite{cormen-introduction} show that the volume of flow across any cut in a flow network is equal to the network flow.

\begin{lemma*}
Let $f$ be a flow in a flow network $R$ with source $s$ and sink $t$, and let $(\mathbb{S}, \mathbb{T})$ be any cut of $R$. The flow across $(\mathbb{S},\mathbb{T})$ is $f(\mathbb{S},\mathbb{T}) = |f|$.
\end{lemma*}

\begin{proof}
See Cormen et.al. \cite{cormen-introduction} Lemma 26.4, p.722.
\end{proof}

An implication of Lemma 26.4 is that the flow of \g into the $NR$ nodes is equal to the \g outflow from the $NS$ nodes. This can be apprehended by evaluating the flow across $(\mathbb{S},\mathbb{T})$ when $s$ and the $NS$ nodes are assigned to $\mathbb{S}$ and all other nodes are assigned to $\mathbb{T}$.

\subsection{Transform the TFN into chains and cycles}
\label{subsec: Transform the TFN into chains and cycles}

The fourth step transforms the TFN into chains and cycles. This is achieved by an algorithm that assigns each element of the flow of \g to a chain or a cycle which - for our purposes - are defined as follows.

\begin{definition*}
A chain is the portion of a simple $s-t$ path (i.e. a path with no repeating nodes) connecting an $NS$ node to an $NR$ node along which a constant positive volume of \g flows. 
\end{definition*}

\begin{definition***}
A cycle is a path along which a constant positive volume of \g flows where
 one node and only one node has inflow and outflow.
\end{definition***}

Chains can be constructed by an algorithm that satisfies the following two conditions; (i) it identifies simple $s-t$ paths with positive flow in the TFN (equivalent to connecting an $NR$ node to an $NS$ node) and (ii) after removing from the network the flow of all the identified simple paths, there is no simple path that has positive flow on each edge. Cycles can be constructed out of the remaining flow in the network by an algorithm that satisfies the following two conditions; (a) it identifies paths with positive flow  starting at one node and ending at that node and (ii) after removing from the network the flows assigned to cycles, there is no cycle that has positive flow on each edge.\footnote{In the TFN Figure \ref{fig:Trade Flow Network} the path $BT_{j}\rightarrow BT_{k} \rightarrow BT_{i} \rightarrow BT_{j}$ is a cycle.} The Flow Decomposition Lemma in Williamson \cite{Williamson-2019} proves that it is possible, in principle, to assign the entire TFN flow capacity to chains and cycles.

\begin{lemma**}
Given an s-t flow f, there exist flows $f_{1},...,f_{\ell},$ for some $\ell \leq E$ such that $f = \sum_{i=1}^{\ell}|f_{i}|$ and for each $i$, the edges of $f_{i}$ with positive flow form either a simple $s-t$ path or a cycle.
\end{lemma**}

\begin{proof}
See Williamson \cite{Williamson-2019} Lemma 2.20, p.48.
\end{proof}

$NR$ and $NS$ nodes cannot be part of a cycle. This follows from the fact that neither $s$ nor \g can be part of a cycle, since the former does not have inflow capacity and the latter does not have outflow capacity.  Therefore, no $NS$ can be part of a cycle, since it only receives flow from $s$ and no $NR$ can be part of a cycle, since it only sends flow to \g. Combining the observations that the outflow of \g from $NS$ nodes is equal to the inflow of \g to $NR$ nodes and that none of those nodes can be part of a cycle, the Flow Decomposition Lemma implies that there is a (possibly non-unique) set of chains connecting $NR$ nodes to $NS$ nodes that contain the entire outflow capacity of the former and the entire inflow capacity of the latter. The remaining flow that is assigned to cycles that involve \g that flows around the interior of the network.

\subsubsection{An algorithm to construct chains and cycles}
\label{subsec: An algorithm to construct chains and cycles}

There are several algorithms that decompose a flow network into chains and cycles. We provide a high-level representation which encompasses any algorithm that can accomplish this task. The algorithm works as follows. First, a sequence of weighted directed edges connecting nodes in the TFN from $s$ to \g, called an $s-t$ path, along which there is a positive flow on each edge is identified. On each $s-t$ path the flow of  the edge with the smallest flow is subtracted from each edge on the path. The resulting graph is the "remainder graph". On the removed path, the edges connecting to $s$ and $t$ are disregarded and the resultant path connecting an $NR$ on one end to an $NS$ node on the other end is a chain. The process is repeated on the remainder graph until there are no more $s-t$ paths with positive flow. Then cycles with positive flows are identified and removed from the graph in a similar manner as the $s-t$ paths. This is repeated until there are no more cycles with positive flow. The Flow Decomposition Lemma ensures that, at the end of the process, there is no flow on any edge of the remainder graph. All the flow in the TFN has been assigned to chains and cycles, each with a constant level flow of \g across each edge.

\begin{algorithm}[t]
\caption{Chain and Cycle Construction Algorithm (Adaptation of Williamson, 2019, Algorithm 2.3)}
\label{alg:chain-cycle}
\Input{Directed trade-flow graph $R=(A,E)$; edge-flow map $f:E\to\mathbb{R}_{\ge 0}$}
\Output{Set of chains $\mathcal{C}$ and set of cycles $\mathcal{Z}$}
$\mathcal{C} \gets \varnothing$\;
$\mathcal{Z} \gets \varnothing$\;

\While{$\operatorname{ExistsSTPathPosFlow}(R)$}{
  $P \gets \operatorname{GetSTPathPosFlow}(R)$ \tcp{an $s$--$t$ path with positive residual flow}
  $f^\star \gets \operatorname{MinFlow}(P)$\;
  \ForEach{$e \in P$}{
    $f[e] \gets f[e] - f^\star$\;
  }
  $\mathcal{C} \gets \mathcal{C} \cup \{(P,f^\star)\}$\;
  $R \gets \operatorname{RemoveChainFlow}(R, P, f^\star)$\;
}

\While{$\operatorname{ExistsFlow}(R)$}{
  $P \gets \operatorname{GetCycleWithFlow}(R)$ \tcp{a directed cycle with positive flow}
  $c^\star \gets \operatorname{MinFlow}(P)$\;
  \ForEach{$e \in P$}{
    $f[e] \gets f[e] - c^\star$\;
  }
  $\mathcal{Z} \gets \mathcal{Z} \cup \{(P,c^\star)\}$\;
  $R \gets \operatorname{RemoveCycleFlow}(R, P, c^\star)$\;
}

\Return{$(\mathcal{C},\mathcal{Z})$}\;
\end{algorithm}

Each node of a chain or cycle - and its inflow and outflow of \g -  is a partition of a node in the TFN, which itself is a partition of a node in the initial  \g-flow graph, Figure \ref{fig:Initial S Flow Network}, which represent agents who enter into trades. As such, the inflow and outflow of \g for each node on a chain or cycle represents a portion of the total   inflow and outflow of \g between an agent and its initial  contract counterparties. We sometimes denote the nodes in chains and cycles as grandchildren of nodes in the initial  \g-flow graph and sometimes simply denote them as agents.

\subsubsection{Computational complexity of the chain and cycle decomposition}
\label{subsubsec: Computational complexity of the chain and cycle decomposition}

Algorithm 2 implies that the edge with the lowest capacity on an augmenting path $P$ is set to zero and  removed from the graph. This suggests a limit on the number of computations that must be made to construct all chains. Let $U = \text{max}_{(u,v)\in R}f(u,v)$ i.e. the largest capacity on any edge in the TFN ; let $m$ equal the minimum unit size of \g, and $E$ represent edges in the TFN.\footnote{For example, if \g represented US Treasuries, the minimum unit is \$100,000.00 of principal.} The next theorem from Williamson \cite{Williamson-2019} states the bound.
\vskip5pt

\begin{theorem***}
\label{th: chain construction}
The Chain Construction Algorithm can be implemented to run in $\mathcal{O}(m log(EU)(A + Alog(A))$ time.
\end{theorem***}

\begin{proof}
See Williamson \cite{Williamson-2019} Theorem 2.24, p. 50.
\end{proof}

\subsection{Attaching \textit{M} to chains and cycles}
\label{subsec:Attaching  M to chains and cycles}

The first four steps decomposed \g flows into chains and cycles. The fifth and final step is to re-attach the initial  contract \m flows to the chains and cycles. For each pair of agents $i$ and $j$, the flow of  \g (say, from $i$ to $j$) is unaltered. The Node Splitting Algorithm divides the flow between (up to) two nodes for each agent. The Chain and Cycle Construction Algorithm further assigns portions of the flow to different chains and cycles. However, neither of these algorithms alter the total flow of \g from $i$ to $j$. As a consequence it is possible to re-attach the flow of \m associated with \g flows. A unit \m flow from $j$ to $i$, $M_{j\to i}$, is derived by dividing the total initial  \m flow into the number of units of \g flow along the edge connecting $i$ to $j$. The unit price is then multiplied by \g flow from $i$ to $j$ on each chain and cycle where they are neighbors.

The flow of \g and the payment assigned to neighboring nodes on a chain or cycle is a fraction of the flow of \g and payment in the initial  contract between their respective grandparents. The \m and \g flows between any pair of agents $i$ and $j$ that enter into an initial  contract for collateral $\g_{ij}$ with unit price $M_{j\to i}$ are assigned to chains and cycles. $\mathbb{CH} = \{1,...,h,...,H\}$ denotes the set of chains generated by the chain construction algorithm and $\mathbb{CY} = \{1,...,y,...,Y\}$ denotes the set of cycles generated by the chain construction algorithm.  The flow of \g on chain $h$ is $T_{h}$.  $\lambda(ij)_{h}$ is the fraction of $T_{ij}$ assigned to chain $h$, such that $\lambda(ij)_{h}T_{ij} =  T_{h}$. Similarly, the flow of \g on cycle $y$ is $T_{y}$.  $\lambda(ij)_{y}$ is the fraction of $T_{ij}$ assigned to cycle $j$, such that $\lambda(ij)_{y}T_{ij} =  T_{y}$. $\sum_{h = 1}^{h = H}\lambda(h)_{ij} + \sum_{y = 1}^{y = Y}\lambda(y)_{ij} = 1$, where $BT_(h)_{ij} = 0$ and $BT_(y)_{ij} = 0$ for any chain or cycle where no transaction between $i$ and $j$ is assigned. The payment of \m from $j$ to $i$ assigned to chain $h$ is  $T_{h}p^{2}_{j\to i}$ and the volume of \g flowing in the opposite direction, from $i$ to $j$, is $S_{h}$. The same applies to the payment of \m from $j$ to $i$ and the reverse flow of \g assigned to cycle $j$. Figure \ref {fig:Assignment of  Trades to Chains and Cycles} is the mapping of the initial  contract trade flows into chains and cycles.

\begin{figure}[H]
$ \underset{\text{initial  contract}}{\underbrace{\{T_{ij}, M_{j\to i}\}}}\;\; \xrightarrow{\hspace*{3cm}} \left\{
\begin{tabular}{p{.8\textwidth}}

$\vdots$

$\lambda(h)_{ij}\cdot T_{ij}, M_{j\to i}\}$  chain $h$

$\vdots$

$\lambda(y)_{ij}\cdot T_{ij}, M_{j\to i}\}$ cycle $j$

$\vdots$

\end{tabular}
\right.$
\caption{Assignment of Trades to Chains and Cycles}
\label{fig:Assignment of Trades to Chains and Cycles}
\end{figure}

On each chain the leftmost node is a child of an $NR$ node, the rightmost node is a child of an $NS$ node and the middle nodes are children of matched-trade $BT$ nodes from the TFN. On each cycle all nodes are children of matched-trade $BT$ nodes from the TFN. An agent $i$ can have as many as two grandchildren nodes per chain and at most one per cycle and may have grandchildren on multiple chains and cycles.\footnote{The path corresponding to a chain could include a $BT$ node representing the same agent as either an $NS$ or $NR$ node. Chain 7 below shows agent $g$ appearing as the $NS$ node and a $BT$ node. An agent cannot appear more than once on a cycle since cycles involve $BT$ nodes and only one node is repeated.} Initial  contract counterparties will be neighbors on chains and cycles and the sum of their trade volume on all chains and cycles on which they are neighbors will equal the volume of trade in their initial   contracts. 

A key property of the decomposed  trade network is that the profit of each agent, defined as its net \m inflow, is unaltered from its value in the initial  contracts reflected in the Initial  \textit{\g}- Flow Graph (Figure \ref{fig:Initial S Flow Network}). Equation \ref{eq:Agent $i$ profit} shows that agent $i$'s net \m inflows obtained by summing over all chains and cycles equals its net \m inflows under its initial  contracts. A similar analysis would show the net inflows of collateral \g to be equal.\footnote{$k$ represents all counterparties of $i$.}

\begin{multline}
\label{eq:Agent $i$ profit}
\sum_{h=1}^{H}\Bigg\{ \underset{\text{$i$'s \m inflow on chain $h$}}{\underbrace{\sum_{k:S_{k\to i > 0}}T_{h}M_{i\to k}}} - \underset{\text{$i$'s \m outflow on chain $h$}}{\underbrace{\sum_{k:T_{i\to k >0}}T_{h}M_{k\to i}}}\Bigg\} +\\
\sum_{y=1}^{Y}\Bigg\{ \underset{\text{$i$'s \m inflow on cycle $j$}}{\underbrace{\sum_{k:T_{k\to i > 0}}S_{h}M_{i\to k}}} - \underset{\text{$i$'s \m outflow on cycle $j$}}{\underbrace{\sum_{k:T_{i\to k >0}}S_{y}p^{2}_{k\to i}}}\Bigg\} =\\
\underset{\text{$i$'s profit from initial  trades}}{\underbrace{\sum_{k:T_{k\to i > 0}}T_{ki}M_{i\to k}- \sum_{k: T_{i\to k > 0}}T_{ik}M_{k\to i}}}
\end{multline}

\subsection{Decomposition of the Trade Flow Network into Chains and Cycles}
\label{subsec:Decomposition of the Synthetic Trade Flow Network into Chains and Cycles}

Here we display one admissible decomposition of the TFN into chains and cycles and associate the corresponding \m-flows, expressed in dollars. The decomposition need not be unique. To construct one such decomposition, we apply a depth-first search routine to identify simple $s-t$ paths with positive flow and then, after those path flows are removed, directed cycles with positive flow; this is consistent with the standard depth-first search procedure in Cormen et.al. \cite{cormen-introduction} Sec. 23.3. The existence of a decomposition of the flow into $s-t$ paths and cycles is guaranteed by the Flow Decomposition Lemma in Williamson (\cite{Williamson-2019}Lemma 2.20, Ch. 2). In the decomposition reported below, \g flows from left to right and \m flows from right to left

{\small{
Chain 1: $NS_{k}$ 
\tikz[baseline]{
  \draw[->] (0,0.1) -- (1,0.1) node[midway, above] {3\g};
  \draw[<-] (0,-0.1) -- (1,-0.1) node[midway, below] {18.9\m};} $NR_{i}$

\vspace{1em}

Chain 2: $NS_{K}$ 
\tikz[baseline]{
  \draw[->] (0,0.1) -- (1,0.1) node[midway, above] {8\g};
  \draw[<-] (0,-0.1) -- (1,-0.1) node[midway, below] {30.16\m};} $BT_{g}$ 
\tikz[baseline]{
  \draw[->] (0,0.1) -- (1,0.1) node[midway, above] {8\g};
  \draw[<-] (0,-0.1) -- (1,-0.1) node[midway, below] {47.6\m};} $NR_{j}$

\vspace{1em}

Chain 3: $NS_{l}$ 
\tikz[baseline]{
  \draw[->] (0,0.1) -- (1,0.1) node[midway, above] {2\g};
  \draw[<-] (0,-0.1) -- (1,-0.1) node[midway, below] {11.9\m};} $BT_{g}$ 
\tikz[baseline]{
  \draw[->] (0,0.1) -- (1,0.1) node[midway, above] {2\g};
  \draw[<-] (0,-0.1) -- (1,-0.1) node[midway, below] {11.9\m};} $NR_{j}$

\vspace{1em}

Chain 4: $NS_{l}$ 
\tikz[baseline]{
  \draw[->] (0,0.1) -- (1,0.1) node[midway, above] {4\g};
  \draw[<-] (0,-0.1) -- (1,-0.1) node[midway, below] {23.8\m};} $BT_{g}$ 
\tikz[baseline]{
  \draw[->] (0,0.1) -- (1,0.1) node[midway, above] {4\g};
  \draw[<-] (0,-0.1) -- (1,-0.1) node[midway, below] {26.12\m};} $NR_{f}$

\vspace{1em}

Chain 5: $NS_{h}$ 
\tikz[baseline]{
  \draw[->] (0,0.1) -- (1,0.1) node[midway, above] {5\g};
  \draw[<-] (0,-0.1) -- (1,-0.1) node[midway, below] {26.25\m};} $BT_{i}$ 
\tikz[baseline]{
  \draw[->] (0,0.1) -- (1,0.1) node[midway, above] {5\g};
  \draw[<-] (0,-0.1) -- (1,-0.1) node[midway, below] {32.75\m};} $NR_{j}$

\vspace{1em}

Chain 6: $NS_{h}$ 
\tikz[baseline]{
  \draw[->] (0,0.1) -- (1,0.1) node[midway, above] {2\g};
  \draw[<-] (0,-0.1) -- (1,-0.1) node[midway, below] {8.2\m};} $BT_{f}$ 
\tikz[baseline]{
  \draw[->] (0,0.1) -- (1,0.1) node[midway, above] {2\g};
  \draw[<-] (0,-0.1) -- (1,-0.1) node[midway, below] {10.24\m};} $NR_{i}$

\vspace{1em}

Chain 7: $NS_{g}$ 
\tikz[baseline]{
  \draw[->] (0,0.1) -- (1,0.1) node[midway, above] {2\g};
  \draw[<-] (0,-0.1) -- (1,-0.1) node[midway, below] {13.06\m};} $BT_{f}$ 
\tikz[baseline]{
  \draw[->] (0,0.1) -- (1,0.1) node[midway, above] {2\g};
  \draw[<-] (0,-0.1) -- (1,-0.1) node[midway, below] {10.24\m};} $BT_{i}$ 
\tikz[baseline]{
  \draw[->] (0,0.1) -- (1,0.1) node[midway, above] {2\g};
  \draw[<-] (0,-0.1) -- (1,-0.1) node[midway, below] {6.0\m};} $BT_{g}$ 
\tikz[baseline]{
  \draw[->] (0,0.1) -- (1,0.1) node[midway, above] {2\g};
  \draw[<-] (0,-0.1) -- (1,-0.1) node[midway, below] {13.06\m};} $NR_{f}$
}}

\vspace{2em}
\noindent
Cycle 1:
\begin{center}
\begin{tikzpicture}[>=stealth]

  \node (g) at (0, 0) {$BT_{g}$};
  \node (f) at (3.5, 0) {$BT_{f}$};
  \node (i) at (1.75, 3) {$BT_{i}$};

  \draw[<->] (g) -- (f) 
    node[midway, above, sloped] {2\g}
    node[midway, below, sloped] {13.06\m};

  \draw[<->] (f) -- (i) 
    node[midway, above, sloped] {2\g}
    node[midway, below, sloped] {10.24\m};

  \draw[<->] (i) -- (g) 
    node[midway, above, sloped] {2\g}
    node[midway, below, sloped] {6\m};

\end{tikzpicture}
\end{center}

\subsection{Maximal multilateral netting}
\label{subsec:Maximal multilateral netting}

Proposition states the key properties achieved by the netting protocol.
\vskip5pt
\begin{proposition}[Maximal multilateral netting]
Suppose Algorithm~1 assigns, at each parent node $i$, the $|\gamma_i|$ units of excess $T$-flow to incident trades in ascending order of the associated $M$-per-$T$ price. Then TradeMech achieves maximal multilateral netting of the designated netted object $T$. More precisely:
\begin{enumerate}
    \item after the decomposition into chains and cycles and the formation of replacement contracts, every matched unit of $T$ is netted multilaterally; the only $T$ that remains un-netted is the unavoidable aggregate imbalance represented by the excess-outflow nodes $NS$ and excess-inflow nodes $NR$; and
    \item among all assignments of excess $T$-flow that satisfy part (1), TradeMech minimizes residual $M$, where residual $M$ denotes the $M$-flow attached to the excess-node assignments.
\end{enumerate}
\end{proposition}

\begin{proof}
By construction, each balanced node $BT$ in the Trade Flow Network satisfies flow conservation, while each excess-outflow node $NS$ carries only excess net outflow and each excess-inflow node $NR$ carries only excess net inflow. By the equal-flow-across-cuts property, the total flow from $s$ to $t$ equals the aggregate excess outflow and, equivalently, the aggregate excess inflow. By the Flow Decomposition Lemma, every positive flow in the Trade Flow Network can be written as a collection of simple $s$--$t$ paths and cycles. Removing the artificial edges incident to $s$ and $t$ converts each $s$--$t$ path into a chain from an $NS$ node to an $NR$ node, while the remaining cycles lie entirely among $BT$ nodes. It follows that all matched $T$-flow is absorbed into chains and cycles. The only $T$ that remains after multilateral netting is therefore the unavoidable excess represented by the chain end-nodes. Any further reduction of that residual $T$ would alter at least one agent's net $T$-position and would violate flow conservation. This proves part (1).

For part (2), fix a parent node $i$ with imbalance $E = |\gamma_i| > 0$. Any admissible split that preserves part (1) must assign exactly $E$ units of incident $T$-flow to $i$'s excess node. Let $x_e$ denote the amount taken from incident edge $e$, and let $p_e$ denote the associated $M$-per-$T$ price on that edge. The residual $M$ contributed by node $i$ is

\[\sum_e p_e x_e\]

subject to

\[\sum_e x_e = E, \qquad 0 \le x_e \le T_e.\]

This linear minimization problem is solved by selecting units in ascending order of $p_e$, which is exactly the rule imposed by Algorithm~1. Because the magnitude of excess $T$ at each node is fixed by the initial netted trade graph, and because total residual $M$ is the sum of the node-level residual $M$ terms, local minimization at each node implies global minimization. Hence, conditional on maximal netting of $T$, the mechanism minimizes residual $M$.
\end{proof}

%% file: preserving_counterparty_risk.tex
\section{Preserving Counterparty Risk}
\label{sec:preserving_counterparty_risk}

In this section we do two things. First, we describe the formation of multiparty contracts that replace the initial contracts and net trades on chains and cycles (the ``netting groups''). Initial  bi-lateral contracts are terminated, portions of trades are assigned to chains and cycles in accordance with the Chain and Cycle Construction Algorithm, and new multiparty replacement contracts  are assigned that net object flows on the chains and cycles. Then we apply the NetWrap \citep{Aronoffnetwrap} method to agent deficiencies in committing required objects.\footnote{NetWrap \citep{Aronoffnetwrap} is a protocol wrapper that enables multilateral netting while preserving the counterparty-risk and profit of the original trades.} 

NetWrap requires agents to send their objects to an escrow where they are locked until execution of the netting trade. When an agent fails to fully fund its obligation, the deficiency-linked portion of its trade assignments are removed from the netting group (the remainder is locked in escrow) and the removed portions are reinstated as bilateral contracts between the original counterparties. The remaining assigned trades are re-netted, which may require the deficient agent or its assignment counterparty to send additional objects to escrow. The process iterates until there is a netting group for which all required objects are locked in escrow. When that is achieved, the netting group is ``executable'' and objects are released from escrow and distributed to agents. This process of replacement contract, escrow, ejection and re-netting has three key properties.

\begin{enumerate}
\item Only the deficient agents and their counterparties are affected by a default. The netted volume of an agent that is neither deficient, nor a counterparty of a deficient agent, is unaffected. 

\item The method preserves initial counterparty risk exposures. An agent is only exposed to a default by its initial counterparties.

\item When there is a minimum size unit of the traded objects, an executable netting group is determined in a finite number of iterations.\footnote{The intuition is that the minimum size of tradeable object sets an upper bound on the number of iterations for which an object is  deficient. For a formal proof, see \citep{Aronoffnetwrap}}
\end{enumerate}

\subsection{Replacement contracts}

Agents enter into initial contracts. After assigning trades to chains and cycles the contracts are terminated and replaced by \mpc{s} that net object flows on each chain and cycle. For each node, the net flow is computed from its connected edges. It is the outflow minus the inflow. A node sends the object for which it has a net outflow and receives the object for which it has an net inflow. Tables \ref{tab:Net flows on Chain 7} and \ref{tab:Net flows on Cycle 1} display he net flow of objects on a chain and a cycle.

\begin{table}[H]
\centering\
\small{
\begin{tabular}{l|c c c c c|c}
\textbf{Object flow} &\textbf{$NS_{g}$} & \textbf{$BT_{f}$}  & \textbf{$BT_{i}$} & \textbf{$BT_{g}$}  & \textbf{$NR_{f}$} & \makecell{Net flow\\ on Chain}\\
\hline
\g - flow   & 2 - out  & &  &   & 2 -in & 0\\
\m - flow   & 13.06 - in & 2.82 - out & 4.24 - out  & 7.06 - in  &  13.06 - out & 0\\              
\end{tabular}
}
\caption{Net flows on Chain 7}
\label{tab:Net flows on Chain 7}
\end{table}
\begin{table}[H]
\centering\
\small{
\begin{tabular}{l|c c c|c}
\textbf{Object flow} &\textbf{$BT_{i}$} & \textbf{$BT_{f}$}  & \textbf{$BT_{g}$}& \makecell{Net flow\\ on Cycle}\\
\hline
\g - flow   &   & & & 0  \\
\m - flow   & 4.24 - out &  2.82 - out & 7.06 - in & 0 \\              
\end{tabular}
}
\caption{Net flows on Cycle 1}
\label{tab:Net flows on Cycle 1}
\end{table}

There are several notable properties of the netting of flows on chains and cycles.

\begin{itemize}
    \item Total net inflow equals total net outflow. This follows from the fact that along each edge the outflow from one node is the inflow to the node at the opposite end.

    \item The only sender of \g on a chain is the left end-node and the only receiver is the right end-node. This follows from the fact that the flow of \g is the same on all edges.

    \item The left end-node of a chain sends \g and receives \m. The right end-node sends \m and receives \g. The middle nodes can only send or receive \m. 

    \item There is no \g sent or received on a cycle.
    
\end{itemize}

\textbf{Verification and locking} nodes are required to lock the objects they are required to send. 

\textbf{Deficiencies} Following the NetWrap method, when a node fails to send its required object the trade assignment whereby the offending node sends the object is removed from the chain or cycle and the remaining trade assignments are re-netted, which results in one or more chains.\footnote{See Aronoff \cite{Aronoffnetwrap} for a description of how to limit removal to the amount of the deficiency and how to allocate a deficient agent's objects between netting groups. We here focus on the simplest case where the deficiency is concentrated on a single netting group and the entire trade assignment is removed.} Visually this corresponds to a removal of the edge connecting the offending node to its neighbor on the side toward which the object flows; the left neighbor in the case of failure to send \m and the right neighbor in the case of failure to send \g.  Below is the decomposition of Chain 7 induced by a failure of $BT_{i}$ to send its net obligations of 6.24 units of \m. The underlying trade between $\{BD_{j}, BD_{i}\}$, which is the initial contract trade between $i$ and $j$, scaled to the percentage of \g assigned to Chain 7, is recovered. No other nodes are affected. The flows of \g and \m along edges are unchanged and the net object flows for all other nodes remain the same as in Chain 7. The recovery of the initial trade between counterparties - scaled to the volume of \g assigned to the chain - reflects that counterparty risk is unaffected by the rearrangement of trades.

{\small{

\[
\text{Chain 7: } NS_{g}\;
\tikz[baseline]{
  \draw[->] (0,0.1) -- (1,0.1) node[midway, above] {$2\g$};
  \draw[<-] (0,-0.1) -- (1,-0.1) node[midway, below] {$13.06\m$};
}
\; BT_{f}\;
\underset{\text{recovered trade}}{%
  \underbrace{%
    \Big[
    \tikz[baseline]{
      \draw[->] (0,0.1) -- (1,0.1) node[midway, above] {$2\g$};
      \draw[<-] (0,-0.1) -- (1,-0.1) node[midway, below] {$10.24\m$};
    }
    \Big]
  }%
}
\; BT_{i}\;
\tikz[baseline]{
  \draw[->] (0,0.1) -- (1,0.1) node[midway, above] {$2\g$};
  \draw[<-] (0,-0.1) -- (1,-0.1) node[midway, below] {$6.0\m$};
}
\; BT_{g}\;
\tikz[baseline]{
  \draw[->] (0,0.1) -- (1,0.1) node[midway, above] {$2\g$};
  \draw[<-] (0,-0.1) -- (1,-0.1) node[midway, below] {$13.06\m$};
}
\; NR_{f}
\]

Chain 7a: 
$NS_{g}$
\tikz[baseline]{
  \draw[->] (0,0.1) -- (1,0.1) node[midway, above] {2\g};
  \draw[<-] (0,-0.1) -- (1,-0.1) node[midway, below] {13.06\m};
}
$BT_{f}$

Chain 7b: 
$BT_{i}$
\tikz[baseline]{
  \draw[->] (0,0.1) -- (1,0.1) node[midway, above] {2\g};
  \draw[<-] (0,-0.1) -- (1,-0.1) node[midway, below] {6.0\m};
}
$BT_{g}$
\tikz[baseline]{
  \draw[->] (0,0.1) -- (1,0.1) node[midway, above] {2\g};
  \draw[<-] (0,-0.1) -- (1,-0.1) node[midway, below] {13.06\m};} $NR_{f}$

Chain 7c (recovered trade): 
$BT_{f}$
\tikz[baseline]{
  \draw[->] (0,0.1) -- (1,0.1) node[midway, above] {2\g};
  \draw[<-] (0,-0.1) -- (1,-0.1) node[midway, below] {10.24\m};
}
$BT_{i}$
}}

The example illustrates a general locality property of deficiency recovery: only the removed assignment is bilaterally reinstated, while the remaining assigned trades are merely repartitioned into residual chains. We state this formally in the following lemma.
\vskip5pt

\begin{lemma}[Locality of deficiency recovery]
Let $G$ be a chain or cycle produced by the TradeMech decomposition, and suppose node $v$ is deficient on the assigned trade $e$ on which $v$ is required to send its net object. Let $G'$ be the residual graph obtained by removing $e$ from $G$, reinstating $e$ as a bilateral contract between the same two original counterparties, and decomposing each nonempty connected component of $G'$ into a residual chain. Then:

\begin{enumerate}
    \item the removed assignment $e$ is recovered as a bilateral contract between the same two original counterparties as in the underlying initial contract from which $e$ was derived;
    \item every surviving assigned trade in the residual chains is a surviving fraction of the same initial bilateral trade from which it was derived before the deficiency;
    \item any node not incident to the removed assignment $e$ retains the same net contractual obligations on the residual chains as it had on $G$; and
    \item no new counterparty exposure is created for any node not incident to $e$.
\end{enumerate}
\end{lemma}

\begin{proof}
By construction, each assigned trade on a chain or cycle is a fraction of an initial bilateral trade between the same two agents. When node $v$ is deficient, the mechanism removes exactly the assigned trade $e$ on which $v$ was required to send its net object. The removed assignment is then reinstated bilaterally between the same two original counterparties, which proves (1).

The mechanism does not alter the identity or size of any surviving assigned trade other than deleting $e$ from the affected chain or cycle. It only re-partitions the remaining connected components into one or more residual chains. Therefore every surviving assigned trade remains a fraction of the same initial bilateral trade as before, proving (2).

Now consider any node not incident to $e$. Since none of its incident assigned trades are changed or removed, the multiset of its surviving inflows and outflows on the residual chains is identical to its inflows and outflows on $G$. Hence its net contractual obligations are unchanged, which proves (3).

Finally, because the only removed-and-recovered trade is reinstated between the same original counterparties, and because all surviving assigned trades preserve their original bilateral ancestry, no node not incident to $e$ acquires a new counterparty. This proves (4). In particular, the recovery procedure changes only the contractual form of the removed assignment and the partition of the surviving assignments into residual chains; it does not alter the surviving assigned $M$-flows.
\end{proof}

\textbf{Margin payments} When the \ic represents a forward contract, contracting parties are subject to margin requirements which are periodic additions or subtractions from the escrow accounts of agents associated with the end-nodes. The margin obligations are typically based on changes in the value at risk ("VAR"). Denoting \g as a security and \m as money for this purpose, VAR is composed of the changes in market price of \g and market volatility from the immediate prior margin adjustment. When VAR increases, the net escrow requirement 
increases for the seller (and correspondingly decreases for the buyer). The opposite occurs when VAR decreases. Subject to uniform margining formulas, margin escrow on a chain is only paid by end-nodes, since the net flow of \g is zero for all intermediate nodes. There is no margin requirement on cycles, since the flow of \g is netted to zero for every node.

\textbf{Iteration} The protocol is applied to each new chain and repeated until there are no deficiencies. 

\textbf{Execution} Once a chain has no deficiencies the \mpc is executed and objects are sent to the receiving parties.

\subsection{Replacement Contract Invariance}
\label{subsec:Replacement Contract Invariance}

Proposition 2 states two key properties of the \tm. One is that agent profit is unaffected by the re-arrangement of contracts. The other is that the \mo protocol does not alter the counterparty exposure.
\vskip5pt

\begin{proposition}[Replacement-contract invariance]
At every stage of the mechanism, including after any sequence of deficiencies and recoveries under Algorithm~3, the following properties hold:
\begin{enumerate}
    \item for each agent, contractual profit, defined as net $M$-inflow, is equal to its net $M$-inflow under the initial contracts; and
    \item each agent's counterparty exposure remains with initial contract counterparties only. In particular, deficiency recovery is local: the removed assignment is reinstated bilaterally between the same two original counterparties, and any node not incident to the removed assignment acquires no new counterparty exposure and retains the same net contractual obligations on the residual chains.
\end{enumerate}
Thus, TradeMech preserves the location of counterparty risk.
\end{proposition}

\begin{proof}
For the initial replacement stage, part (1) follows from Equation~(2), which shows that the sum of an agent's net $M$-inflows across all chains and cycles equals its net $M$-inflow under the initial contracts. Part (2) holds initially because, by construction, each edge on each chain or cycle is a fraction of an initial bilateral trade between the same two agents.

Now suppose the proposition holds after $k$ deficiency events. Consider event $k+1$. The deficiency removes exactly one assigned trade from an affected chain or cycle and triggers the recovery procedure described in Section~3.1. By the Locality of Deficiency Recovery Lemma, the removed assignment is reinstated bilaterally between the same two original counterparties, every surviving assigned trade remains a surviving fraction of the same initial bilateral trade as before, and any node not incident to the removed assignment retains the same net contractual obligations and acquires no new counterparty exposure.

It follows that the mechanism preserves the identity of each agent's counterparties: no new counterparty pair is introduced, and all surviving or recovered obligations remain tied to original bilateral counterparties. Moreover, since the only change is the contractual form of the removed assignment and the repartition of the remaining assignments into residual chains, each agent's total net $M$-inflow remains unchanged. Therefore parts (1) and (2) continue to hold after event $k+1$. By induction, they hold after any sequence of deficiencies and recoveries.
\end{proof}

Proposition 2 preserves the identity of the counterparty bearing default risk. It does not require the gross magnitude of counterparty risk to remain unchanged. That magnitude may fall under the replacement contracts because direct netting of $T$ on chains eliminates pass-through exposure at intermediate nodes.

\subsection{Algorithm for Multiparty Contractual Obligations}
\label{subsec: Algorithm for Multiparty Contractual Obligations}

A netting group is a chain or cycle produced by Algorithm~2 together with its assigned $T$-flows, assigned $M$-flows, and replacement contract.

For a node $v$ in a netting group $G$, let $d_G(v)$ denote the delivery edge of $v$: the unique incident assigned trade on which $v$ is required to deliver its net object under the replacement contract for $G$. On a chain, $d_G(v)$ is the edge to the right if $v$ must send $T$, and the edge to the left if $v$ must send $M$. On a cycle, where $T$ nets to zero, $d_G(v)$ is the unique incident edge on which $v$ is designated to deliver its net $M$-obligation.

\begin{boxH}
\textbf{Algorithm 3: Replacement Procedure for Multiparty Contractual Obligations.}\\
\textbf{Input:} initial family $\mathcal{G}_0$ of chains and cycles produced by Algorithm~2.\\
\textbf{Output:} executable netting groups $\mathcal{X}$ and recovered bilateral contracts $\mathcal{B}$.

\begin{enumerate}
    \item Initialize the working set of netting groups by setting
    \[
    \mathcal{Q} \gets \mathcal{G}_0,
    \qquad
    \mathcal{X} \gets \varnothing,
    \qquad
    \mathcal{B} \gets \varnothing.
    \]

    \item While $\mathcal{Q} \neq \varnothing$, remove one netting group $G$ from $\mathcal{Q}$ and proceed as follows.

    \begin{enumerate}
        \item If $G$ is executable, add $G$ to $\mathcal{X}$.

        \item Otherwise, choose one deficient node $v$ in $G$, let
        \[
        e \gets d_G(v),
        \]
        and recover the assigned trade $e$ bilaterally between the same original counterparties. Add that recovered bilateral contract to $\mathcal{B}$.

        \item Form the residual graph
        \[
        G^{-} \gets G \setminus \{e\}.
        \]

        \item For each connected component $H$ of $G^{-}$ that contains at least one edge, treat $H$ with its inherited edge order and inherited assigned flows as a residual chain, and add $H$ to $\mathcal{Q}$.
    \end{enumerate}

    \item When $\mathcal{Q}$ is empty, return $(\mathcal{X}, \mathcal{B})$.
\end{enumerate}
\end{boxH}

Algorithm~3 is a local edge-removal procedure. Each deficiency removes exactly one assigned trade: the trade on which the deficient node was required to deliver its net object. That removed trade is recovered bilaterally between the same original counterparties, while the remaining assigned trades are unchanged and are merely repartitioned into residual chains. Deleting one edge from a chain yields at most two residual chains; deleting one edge from a cycle yields one residual chain. The procedure is then repeated on the residual chains until all remaining netting groups are executable.

%% file: additional_applications.tex
\section{Single Traded Object}
\label{sec:Single_Traded_Object}

When there is only one object that is traded between counterparties, for example in a derivatives trade where only one party has a margin obligation, the fourth column from Table \ref{tab:Initial Contracts} is removed, which eliminates object $\m$. The protocol remains the same as before except that there is no re-attachment of \m to chains and cycles.\footnote{The restriction of flow to a single object eliminates obligations to deliver objects on cycles.} As an example, Chain 7 becomes:

{\small{
Chain 7: $NS_{g}$ 
\tikz[baseline]{
  \draw[->] (0,0.1) -- (1,0.1) node[midway, above] {2\g};}
  $BT_{f}$ 
\tikz[baseline]{
  \draw[->] (0,0.1) -- (1,0.1) node[midway, above] {2\g};}
  $BT_{i}$ 
\tikz[baseline]{
  \draw[->] (0,0.1) -- (1,0.1) node[midway, above] {2\g};}
  $BT_{g}$ 
\tikz[baseline]{
  \draw[->] (0,0.1) -- (1,0.1) node[midway, above] {2\g};}
  $NR_{f}$
}}

Netting results in $NS_{g}$ being the only node that contributes the object \g, and node $NR_{f}$ is the only receiver of the object. In aggregate, it is only the net sending amounts that are required to be sent. The algorithm for decomposing multiparty contracts in Section \ref{subsec: Algorithm for Multiparty Contractual Obligations} continues to apply. a failure to send \g by node $NS_{g}$ causes the chain to decompose into a bilateral contract for the delivery of 2 units of \g from $g$ to $f$.\footnote{The remainder of the chain becomes a cycle where there is no required delivery of \g.}

%% file: conclusion.tex
\section{Conclusion}
\label{sec: Conclusion}

Existing methods used to multilaterally net trades in a market alter counterparty risk. Trade compression can alter counterparties, and central clearing imposes a CCP as a new counterparty to every trade. We presented a trading mechanism that achieves maximal multilateral netting without altering counterparty credit risk. The mechanism does so by assigning trades to chains and cycles and replacing initial bilateral contracts with multilateral contracts on each chain and cycle. Trades assigned to chains and cycles are netted, which achieves maximal multilateral netting of the designated netted object. At the same time, the replacement contracts preserve each agent’s contractual profit and preserve the location of counterparty risk, because all assigned trades remain fractions of initial bilateral trades between the same counterparties.

When a party fails to pre-commit the objects it is required to send, the affected trade assignment is removed and recovered as a bilateral contract between the same original counterparties. The remaining trade assignments are re-netted on re-formed chains, and no new counterparty exposure is created. A key insight is that the transformation into chains and cycles is achieved without altering the edge connections between agents, either in terms of counterparties or in terms of the aggregate volume of traded objects flowing between them. This invariance enables recovery of the initial bilateral obligation when a party fails to perform, thereby preserving counterparty exposures while maintaining multilateral netting of the remaining trades. The same logic extends to the case in which only a single object is traded.